\documentclass[11pt,letterpaper]{article}

\usepackage[utf8]{inputenc}
\usepackage[OT4]{fontenc}
\usepackage{amsthm,amssymb,amsmath}
\usepackage{xspace,enumerate}
\usepackage[utf8]{inputenc}
\usepackage{thmtools}
\usepackage{thm-restate}
\usepackage{authblk}
\usepackage[hypertexnames=false,colorlinks=true,urlcolor=Blue,citecolor=Green,linkcolor=BrickRed]{hyperref}
\usepackage{array, makecell}

\def\poly{\operatorname{poly}}
\def\polylog{\operatorname{polylog}}

  \theoremstyle{plain}
  \newtheorem{theorem}{Theorem}
  \newtheorem{lemma}[theorem]{Lemma}
  \newtheorem{corollary}[theorem]{Corollary}
  
  \newtheorem{fact}[theorem]{Fact}
  \newtheorem{observation}[theorem]{Observation}
  \newtheorem{definition}[theorem]{Definition}

\title{A Deterministic Parallel APSP Algorithm and its Applications}

\date{\vspace{-5ex}}
\author[1]{Adam Karczmarz\thanks{\texttt{a.karczmarz@mimuw.edu.pl}. Supported by ERC Consolidator
Grant 772346 TUgbOAT, the Polish National Science Centre 2018/29/N/ST6/00757 grant, and by the Foundation for Polish Science (FNP) via the START programme.}}

\author[1]{Piotr Sankowski\thanks{\texttt{sank@mimuw.edu.pl}. Supported by ERC Consolidator Grant 772346 TUgbOAT.}}

\affil[1]{Institute of Informatics, University of Warsaw, Poland}

\usepackage[dvipsnames,usenames]{color}

\usepackage[ruled]{algorithm}
\usepackage[noend]{algpseudocode}
\algnewcommand{\LineComment}[1]{\State \(\triangleright\) #1}

\usepackage[margin=1in]{geometry}
\usepackage{graphicx}
\usepackage[noadjust]{cite}
  \usepackage{microtype}
    \usepackage{xspace}
    \usepackage{amsmath,amsfonts}
    \usepackage{wasysym}
  \usepackage{makecell}
  \usepackage{tabularx}
  \usepackage{comment}
  \usepackage{todonotes}
  \usepackage{thmtools}
  \usepackage{thm-restate}
  \usepackage[capitalise]{cleveref}
  \usepackage{verbatim}
\usepackage{units}
\usepackage{color}
\definecolor{darkblue}{rgb}{0,0.08,0.45}
  \usepackage{epsfig}
    \usepackage{graphics}
\usepackage{pgf,tikz}
  \usetikzlibrary{trees, arrows, shapes, snakes}
\usepackage{marvosym}
  \usepackage[caption=false]{subfig}
\usetikzlibrary{decorations.pathmorphing}
\usetikzlibrary{decorations.markings}
\usetikzlibrary{decorations.pathmorphing,shapes}
\usetikzlibrary{calc,decorations.pathmorphing,shapes}
\usepackage{subfig}
\usepackage{booktabs}
\usepackage{etoolbox}
\usepackage[title]{appendix}
\usepackage{enumitem}

\newif\iffull
\fulltrue

\newif\ifspreport
\spreporttrue

\newcommand{\dist}{\delta}

\newcommand{\wei}{w}
\newcommand{\disth}[1]{\ensuremath{\dist^{#1}}}
\newcommand{\len}{\ell}

\newcommand{\hops}[1]{\ensuremath{|#1|}}

\newcommand{\Ot}{\ensuremath{\widetilde{O}}}
\newcommand{\eps}{\ensuremath{\epsilon}}

\newcommand{\DDG}{DDG}

\newcommand{\bnd}{\partial}

\begin{document}

\maketitle

\thispagestyle{empty}

\begin{abstract}

In this paper we show a deterministic parallel all-pairs shortest paths algorithm for real-weighted directed graphs.
  The algorithm has $\Ot(nm+(n/d)^3)$ work and $\Ot(d)$ depth for any depth parameter $d\in [1,n]$.
  To the best of our knowledge, such a trade-off has only
been previously described for the real-weighted single-source shortest paths problem
using randomization [Bringmann et al., ICALP'17]. Moreover, our result improves
upon the parallelism of the state-of-the-art randomized parallel algorithm for computing transitive closure, which has
$\Ot(nm+n^3/d^2)$ work and $\Ot(d)$ depth [Ullman and Yannakakis, SIAM J. Comput. '91].

Our APSP algorithm turns out to be a powerful tool for designing efficient planar graph algorithms
in both parallel and sequential regimes.
  By suitably adjusting the depth parameter~$d$ and applying known techniques,
  we obtain:
  \begin{enumerate}[label=(\arabic*)]
    \item
  nearly work-efficient $\Ot(n^{1/6})$-depth parallel algorithms for the real-weighted single-source shortest paths problem
  and finding a bipartite perfect matching in a planar graph,
    \item an $\Ot(n^{9/8})$-time sequential strongly polynomial algorithm for
computing a minimum mean cycle or a minimum cost-to-time-ratio cycle of a planar graph,
    \item a slightly faster
  algorithm for computing so-called external dense distance graphs of
  all pieces of a recursive decomposition of a planar graph.
  \end{enumerate}

  One notable ingredient of our parallel APSP algorithm is a simple deterministic \linebreak $\Ot(nm)$-work
  $\Ot(d)$-depth procedure for computing $\Ot(n/d)$-size
hitting sets of shortest \linebreak $d$-hop paths between
all pairs of vertices of a real-weighted digraph. Such hitting~sets have also been called $d$-hub sets.
Hub sets have previously proved especially useful in designing parallel or dynamic shortest paths algorithms
and are typically obtained via random sampling.
Our procedure implies, for example, an $\Ot(nm)$-time deterministic algorithm for
finding a shortest negative cycle of a real-weighted digraph.
Such a near-optimal bound for this problem has been so far only achieved
  using a randomized algorithm~[Orlin et al., Discret. Appl. Math. '18].

\end{abstract}

\clearpage
\setcounter{page}{1}

\section{Introduction}

The \emph{all-pairs shortest paths problem} (APSP) is one of the most fundamental
graph problems. It has been studied in numerous variants, for many computational
models and graph classes. In this paper we study the APSP problem on
real-weighted, possibly sparse graphs in the parallel setting.

The efficiency of a parallel algorithm is usually characterized using
the notions of \emph{work} and \emph{depth} (also called~\emph{span}~or~\emph{time}).
The work is the total number of primitive operations performed.
The depth is the longest chain of sequential dependencies between these operations.
A parallel algorithm of work $W(n)$ and depth $D(n)$ (where $n$ is the problem size)
can be generally scheduled
to run in $\Ot(D(n))$ time using $\Ot(W(n)/D(n))$ processors~\cite{Brent74},
where the $\Ot(\cdot)$ notation suppresses $O(\polylog{n})$ factors.
The quantity $W(n)/D(n)$ is often called the \emph{parallelism} of a parallel algorithm.
An algorithm is called \emph{nearly work-efficient} if $W(n)=\Ot(T(n))$, where $T(n)$
is the best known time bound needed to solve the problem using a sequential algorithm.

There exists a very simple folklore parallel algorithm for the APSP problem via repeatedly squaring the weighted adjacency matrix using min-plus product.
It has $\Ot(n^3)$ work and $O(\polylog{n})$ depth. This
can be slightly tweaked to obtain polylogarithmic-factor improvement
in work and depth~\cite{10.1145/140901.141913}.
However, these algorithms are nearly work-efficient only for dense graphs
since the best known sequential algorithms run in $\Ot(nm)$ time~\cite{Johnson77, Pettie04}.
Hence, we are missing APSP algorithms competitive with the state-of-the-art sequential
algorithms even for moderately dense graphs.

Dealing with sparser graphs in the parallel setting
turns out to be a much more challenging task
even for an easier problem of computing the transitive closure.
To the best of our knowledge, the state-of-the-art for parallel transitive
closure for sparse graphs is the classical
tradeoff of Ullman and Yannakakis~\cite{UllmanY91}.
They showed a Monte Carlo randomized parallel algorithm with $\Ot(nm+n^3/d^2)$ work
and $\Ot(d)$ depth for any parameter $d\in [1,n]$.
Hence, their algorithm is nearly work-efficient for $d=\widetilde{\Omega}(n/\sqrt{m})$
and can achieve parallelism of order $\Ot(m^{3/2})$ while being nearly work-efficient.
\subsection{Our Results and Related Work}
Our main result is a \emph{deterministic} parallel all-pairs shortest paths algorithm for \emph{real-weighted}
directed graphs that improves upon the 30-year-old randomized transitive closure trade-off of~\cite{UllmanY91}.

\begin{restatable}{theorem}{tapsp}\label{t:apsp}
  Let $G$ be a real-weighted digraph. For any $d\in [1,n]$, there
  exists a deterministic parallel algorithm computing all-pairs shortest paths in $G$ with
  $\Ot(nm+(n/d)^3)$ work and $\Ot(d)$ depth.
\end{restatable}

Observe that our algorithm is nearly work-efficient for $d=\Omega(n^{2/3}/m^{1/3})$,
which is $\Omega(n^{1/3})$ for sparse graph.
As a result, as long as the number of used processors is $p=\Ot(n^{1/3}m^{4/3})$,
we can compute all-pairs shortest paths in $\Ot(nm/p)$ parallel time.
To the best of our knowledge, such a tradeoff for \emph{real-weighted} digraphs has only
been achieved for \emph{single-source shortest paths} (SSSP) and \emph{negative cycle detection}
problems~\cite{BringmannHK17}. Both of these results require randomization, whereas our algorithm is deterministic.
Bringmann et al.~\cite{BringmannHK17} also show a deterministic variant of their algorithm
with work $\Ot(nmd+(n/d)^3)$ that is not nearly work-efficient unless the graph is dense.\footnote{Bringmann et al.~\cite[Theorem 19]{BringmannHK17} mistakenly
state the work of their algorithm to be $\Ot(nm+(n/d)^3+n^2d)$ --
even though in Section 4.1.2 they correctly bound the work in
the first step of their algorithm to be $\Theta(nmd)$.}

In the aforementioned previous results~\cite{UllmanY91, BringmannHK17}, randomization is
used only for computing a small 
subset of $V$ that is, roughly speaking, guaranteed to contain some
vertex of each ``long'' shortest path consisting of at least $h$ hops.
Following~\cite{KarczmarzL19}, we call such set a $h$-hub set of $G$ (a formal definition
is deferred to Section~\ref{s:apsp})\footnote{Zwick \cite{Zwick02} uses the name \emph{bridging set} for an analogous concept.
Some works also use the term \emph{hitting set}, but hitting set is a more general notion, which in our paper is used in multiple different contexts.}.
A classical argument of Ullman and Yannakakis~\cite{UllmanY91} shows that a randomly sampled
$\Theta((n/h) \cdot \log{n})$-size subset of $V$ constitutes an $h$-hub set with high probability.

We show that for any $h\in [1,n]$, an $h$-hub set of size $\Ot(n/h)$ can be
computed deterministically using $\Ot(nm)$ work and $\Ot(h)$ depth.
Our procedure works in presence of negative edge weights and even -- to some extent -- when negative cycles are allowed.
The constructed hub set is guaranteed to hit all-pairs shortest $h$-hop paths (which are
well-defined regardless of whether the APSP problem is feasible)
unless the shortest (in terms of hops) negative cycle has at most $h$ edges.
As a by-product, we also obtain the following result in the sequential regime.
\begin{restatable}{theorem}{tcycle}\label{t:cycle}
  Let $G$ be a real-weighted directed graph.
  A shortest (in terms of hops) negative-weight cycle in $G$ can be found
  deterministically in $O(nm\log^2{n})$ time.
\end{restatable}
So far, an $\Ot(nm)$ bound for the shortest negative cycle problem has
only been obtained using a randomized algorithm by Orlin et al.~\cite{OrlinSW18}\footnote{The algorithm of~\cite{OrlinSW18} runs in $O(nm\log{n})$ expected time. However,
if one aims at high-probability correctness, its running time is $O(nm\log^2{n})$ which matches our bound.}.
The best known deterministic algorithms~\cite{Subramani09, SubramaniWG13}
require $\widetilde{\Omega}(n^3)$ worst-case time.
One can easily argue that our bound is the best possible, up to polylogarithmic factors, unless
an unexpected algorithmic breakthrough is made, that would imply progress for other core problems as well.
In particular, the shortest negative cycle captures
the unweighted directed girth problem (for which the trivial $O(nm)$ bound stands for decades).
Moreover, one should not hope for an $O(n^{3-\eps})$ fast matrix multiplication-based
algorithm since the shortest negative cycle problem also captures the negative triangle detection problem known
to be subcubic-equivalent to the APSP problem~\cite{WilliamsW18}.

We believe that our procedure for computing hub sets might
be useful in obtaining other deterministic sequential and parallel algorithms.
 As a direct application
of our APSP algorithm one can obtain, e.g.,  more efficient parallel algorithms for computing closeness
centrality~\cite{1690659}. Moreover, as shown below,
we can use it to improve algorithms for several planar graph problems.

\subsubsection{Applications for planar graphs}

\paragraph{Nearly work-efficient parallel algorithms.}
Theorem~\ref{t:apsp} can be used to highly parallelize the framework used
by Fakcharoenphol and Rao~\cite{FR} to obtain
nearly optimal algorithms for two fundamental planar graph problems.

\begin{restatable}{theorem}{tparallel}\label{t:parallel}
  Let $G$ be a real-weighted planar digraph. There exists
  a deterministic parallel algorithm for negative cycle detection
  and single-source shortest paths in $G$ with $\Ot(n)$ work
  and $\Ot(n^{1/6})$ depth.
\end{restatable}

By a well-known duality-based reduction~\cite{MillerN95}, this also implies
\emph{feasible flow} and \emph{bipartite perfect matching} algorithms with the same bounds.
It is worth noting that even though there exist $O(\polylog n)$ depth algorithms (i.e., belonging to NC class) for finding
perfect matchings in bipartite planar graphs, they are very far from
being work-efficient~\cite{MillerN95,10.1145/335305.335346}.
The same applies to the $\Ot(\sqrt{n})$-depth algorithm implied by the
interior-point method-based result for general graphs~\cite{10.1137/0221011}.
Since the $s,t$-max flow problem on planar graphs with capacities $[1,C]$
can be reduced to $O(\log{C})$ feasible flow computations~\cite{MillerN95}, this
also yields a nearly work-efficient $\Ot(n^{1/6})$-depth algorithm
for maximum $s,t$-flow with polynomially bounded capacities.
Similar bounds can be obtained for the related
\emph{replacement paths} problem using the recent reduction~\cite{ChechikN20}
to the all-edge shortest cycles problem (see Section~\ref{s:parallel-sssp} and Appendix~\ref{a:external}).

To the best of our knowledge, the parallel complexity of the aforementioned
problems on planar graphs has not achieved much attention in recent years.
However, one can easily obtain near-optimal work and $\Ot(\sqrt{n})$-depth algorithm
for these problems using the breakthrough framework of Fakcharoenphol and Rao~\cite{FR}.
In this framework, one repeatedly computes all-pairs shortest paths
on certain \emph{dense distance graphs} using multiple runs of a clever implementation
of Dijkstra's algorithm (so-called FR-Dijkstra).
Unfortunately, it is not clear how to break the $\Omega(\sqrt{n})$ depth bound
this way since Dijkstra's algorithm is inherently sequential.
Our improved depth bound is obtained by replacing the simple-minded Dijkstra-based
APSP algorithm
with that of Theorem~\ref{t:apsp}. As we show, the Monge property of
dense distance graphs that is crucial for the efficiency of FR-Dijkstra can
also be employed in our algorithm to yield a significant parallel speed-up.
\paragraph{Minimum mean and cost-to-time ratio cycle problems.} By plugging in our
parallel negative cycle detection algorithm into Megiddo's
\emph{parametric search} framework~\cite{Megiddo83} we obtain improved \emph{strongly polynomial}
algorithms for the minimum mean cycle and minimum cost-to-time ratio cycle\footnote{Also known as the \emph{minimum ratio cycle} problem.} problems on planar graphs
(for formal definitions of these problems, refer to Section~\ref{s:meancycle}).

\begin{restatable}{theorem}{tmeancycle}\label{t:meancycle}
  Let $G$ be a real-weighed planar graph.
  There exists an $\Ot(n^{9/8})$-time strongly polynomial algorithm
  for computing a minimum ratio cycle (and thus also a minimum mean cycle) in $G$.
\end{restatable}
The minimum mean cycle and minimum cost-to-time ratio cycle problems are classical
graph problems studied since the seventies. They are used to construct
strongly polynomial algorithms for computing minimum-cost flows~\cite{Tardos85}.
Moreover, via the cut-cycle duality in planar graphs, both problems have found practical applications in
the area of image segmentation~\cite{JermynI01, Veksler02, WangS01, WangS03}

It is known that both problems can be reduced to negative cycle detection via binary search~\cite{Lawler}.
However, this way we can obtain only weakly polynomial time algorithms with running times
dependent on the magnitude of edge weights.
For general graphs, the classical algorithm of~Karp~\cite{Karp78} solves the minimum mean cycle
problem in $O(nm)$ time, matching the best known strongly polynomial negative cycle
detection bound achieved by the Bellman-Ford algorithm.
Karp's algorithm (and other minimum mean cycle algorithms for general graphs~\cite{HartmannO93, YoungTO91})
operates on limited-hop shortest paths. As a result it is not clear how
to take advantage of planarity~to~speed~it~up\footnote{At-most-$h$-hop shortest paths connecting pairs of vertices of
a single face of a plane graph do not seem to admit
algorithmically useful properties, like the non-crossing property of usual shortest paths.}.

The $\Ot(nm)$ bound has not been matched to date for the more general minimum
cost-to-time ratio problem which seems to be the original inspiration for the
invention of the parametric search technique~\cite{Megiddo83} that later found other applications (e.g.,~\cite{AgarwalST94}).
This technique can be used to convert an efficient parallel negative cycle detection
algorithm into a strongly polynomial minimum ratio cycle algorithm.
The best known strongly polynomial $\Ot(m^{3/4}n^{3/2})$ bound for the minimum ratio cycle
problem  is due to Bringmann et al.~\cite{BringmannHK17}
and also follows by plugging their aforementioned parallel negative cycle detection
algorithm into Megiddo's framework.

It seems that no strongly polynomial algorithms for the minimum mean cycle problem to date have been designed
specifically for planar graphs. However, by plugging in previously
known parallel negative cycle detection algorithms~\cite{FR, LiptonT80}
into the parametric search framework, one would only obtain $\Ot(n^{3/2})$-time strongly
polynomial algorithms.
Theorem~\ref{t:meancycle} improves upon this significantly.
\paragraph{Computing external dense distance graphs.}
Finally, our parallel APSP algorithm can be used to improve the bound for
computing so-called \emph{external dense distance graphs} wrt. a planar graph's recursive decomposition -- an important black box
with applications in computing maximum flows~\cite{LackiNSW12}, minimum cuts~\cite{BorradaileSW15}, and constructing
distance oracles~\cite{exact_oracle, failing_oracle, MozesS12}.

\newcommand{\TG}{\mathcal{T}}

Suppose a planar graph $G$ is recursively
decomposed using small cycle separators~\cite{Miller86} of size $O(\sqrt{n})$
until the obtained pieces have constant size.
The decomposition procedure produces a binary tree
$\TG(G)$ whose nodes correspond to subgraphs of $G$ (\emph{pieces}).
The \emph{boundary vertices} $\bnd{H}$ of a piece $H\in\TG(G)$ are vertices that $H$ shares
with the remaining part $G-H$ of the entire graph~$G$.
We denote by $\DDG_H$ the \emph{dense distance graph} -- a complete
weighted graph on $\bnd{H}$ whose edge weights represent distances
between all pairs of vertices of $\bnd{H}$ in~$H$.
Efficient construction of piecewise DDGs and FR-Dijkstra alone
are enough to obtain e.g., nearly linear-space static and dynamic exact distance
oracles with sublinear query time~\cite{FR, KaplanMNS17, MSSP}.
Originally, Fakcharoenphol and Rao~\cite{FR} gave an $O(n\log^3{n})$
algorithm for computing each of the graphs $\DDG_H$, $H\in \TG(G)$, inductively,
based on the children graphs $\DDG_{H_1},\DDG_{H_2}$.
The key ingredient in their algorithm was the aforementioned
fast implementation of Dijkstra's algorithm on a dense distance graph.

Later, Klein~\cite{MSSP} showed a more efficient,
$O(n\log^2{n})$-time
algorithm that computed every $\DDG_H$ directly, by building the so-called multiple-source shortest paths (MSSP)
data structure for each piece $H\in\TG(G)$ separately.
This was possible since the total size of all pieces $H$ of the decomposition is $O(n\log{n})$.
However, some important applications~\cite{exact_oracle, failing_oracle, BorradaileSW15, LackiNSW12, MozesS12} also require \emph{external} graphs
$\DDG_{G-H}$ representing distances between $\bnd{H}$ in the complement graph $G-H$ (for all $H\in\TG(G)$).
In this case Klein's approach fails since the total size of all possible
graphs $G-H$ is $\Omega(n^2)$.
Instead, one has to stick to the original inductive method of~\cite{FR}
and compute each $\DDG_{G-H_i}$ based on $\DDG_{G-H}$ and $\DDG_{H_{3-i}}$
(where, again, $H$ is a parent of $H_1,H_2$ in $\TG(G)$).
This takes $O(n\log^3{n})$ time through all $H$.
We show that our parallel APSP algorithm can be used to obtain an algorithm whose running time
almost matches the $O(n\log^2{n})$ MSSP-based bound.

\begin{restatable}{lemma}{lexternal}\label{l:external}
  The external dense distance graphs $\DDG_{G-H}$ for all $H\in \TG(G)$
  can be computed in \linebreak $O(n\log^2{n}\cdot \log\log{n}\cdot \alpha(n))$ time.
\end{restatable}

\subsection{Further Related Work}

For general directed graphs, the parallel \emph{single-source shortest paths} problem (SSSP) is much
better understood. Ullman and Yannakakis~\cite{UllmanY91} showed an $\Ot(m\sqrt{n})$-work
algorithm with $\Ot(\sqrt{n})$ depth for unweighted digraphs.
Besides the aforementioned result of Bringmann et al.~\cite{BringmannHK17}
and $\widetilde{\Theta}(n)$-depth parallel implementations of
Dijkstra's algorithm~\cite{BrodalTZ98},
all known parallel exact SSSP algorithms work for weighted graphs
with \emph{non-negative} and \emph{integral} weights bounded by $W$ and are weakly polynomial.
Klein and Subramanian~\cite{KleinS97} generalized the bound of~\cite{UllmanY91} to the weighted setting
at the cost of $\Ot(\log{W})$ factor in work/depth bounds.
In this setting, they also gave a nearly work-efficient polylog-depth
algorithm for planar graphs~\cite{KleinS93}.
Spencer~\cite{Spencer97} gave a trade-off algorithm with
$\Ot((n^3/d^2+m)\log{W})$ work and $\Ot(d)$ depth.
Forster and Nanongkai~\cite{ForsterN18} in turn have recently shown a trade-off algorithm
with $\Ot((md+mn/d+(n/d)^3)\log{W})$ work and $\Ot(d)$ depth.

Recently, following similar results for single-source reachability~\cite{Fineman18, LiuJS19},
Cao et al.~\cite{CaoFR20} showed that $(1+\eps)$-approximate single-source
shortest paths can be found using near-optimal work $\Ot(m\log{W})$ and $\Ot(n^{1/2+o(1)})$
depth. Note that this SSSP algorithm, if run from every source, yields a nearly work-efficient
(approximate) APSP algorithm, but with polynomially worse depth than ours.

Finally, parallel single-source shortest paths problems in
\emph{undirected} graphs have also received much attention, both in
the exact~\cite{Blelloch0ST16, Cohen97, ShiS99} and the approximate~\cite{AndoniSZ20, Cohen00, Li20}
setting.

\subsection{Technical Overview}
Our parallel APSP algorithm is based on the techniques used in the state-of-the-art Monte Carlo randomized
decremental all-pairs shortest paths algorithm for weighted digraphs, due to Bernstein~\cite{Bernstein16}.
Without loss of generality assume that our depth parameter is a power of two, i.e., $d=2^K$.
Our algorithm makes use of a \emph{hierarchy} of hub sets
$V=H_1,H_2,\ldots,H_{2^i},\ldots, H_d\subseteq V$ such that each
$H_{2^i}$ is a $2^i$-hub set of~$G$ and has size $O((n/2^i)\log{n})$.
Roughly speaking, this means that for all pairs $u,v\in V$, some shortest $u\to v$
path, if it consists of at least $2^i$ edges,
contains a vertex of~$H_{2^i}$.

Let us start with describing a randomized version of our algorithm.
A well-known fact attributed to Ullman and Yannakakis~\cite{UllmanY91} states that
picking each $H_{2^i}$ to be a random $\Theta((n/2^i)\cdot\log{n})$-subset of $V$
guarantees that $H_{2^i}$ forms a $2^i$-hub-set of $G$ with high probability.

Rather than using the inherently sequential Dijkstra's algorithm, as sequential $\Ot(nm)$-time APSP algorithms~\cite{Johnson77, Pettie04}
do, we rely exclusively on the Bellman-Ford algorithm.
A variant of Bellman-Ford algorithm, given a source $s$, maintains \emph{distance labels} and performs a number of \emph{steps} relaxing
all edges in arbitrary order in $O(m)$ time.
After $k$ steps, the distance label of $v\in V$ is equal to the length $\dist_G^k(s,v)$ of a shortest
at-most-$k$-hop path from $s$ to $v$.
In a single step, for each vertex we need to combine edge relaxations ending in this vertex,
so a single step requires $O(\polylog{n})$ depth.
Consequently, performing $k$ steps of Bellman-Ford requires $O(mk)$ work but only $\Ot(k)$ depth.

The key idea behind Bernstein's algorithm~\cite{Bernstein16} is that
computing shortest paths $\dist_G(u,v)$ for $(u,v)\in H_{2^i}\times V$ can be reduced to
computing \emph{at-most-$2^{i+1}$-hop} shortest paths
on a graph $G_{i,s}$ obtained from $G$ by
augmenting it with $|H_{2^{i+1}}|=O(n)$ auxiliary edges $st$ for all $t\in H_{2^{i+1}}$,
such that the weight of $st$ equals the distance $\dist_G(s,t)$ from $s$ to $t$.
This idea suggests an inductive algorithm
that, given distances for all pairs in
$(H_{2^{i+1}}\times V)\cup (V\times H_{2^{i+1}})$, where $i<K$,
lifts them to distances for all pairs in $(H_{2^i}\times V)\cup(V\times H_{2^i})$
using $|H_{2^i}|$ parallel $2^{i+1}$-step Bellman-Ford
runs. These can be performed in $O(|H_{2^i}|\cdot 2^i\cdot m)=O(nm\log{n})$ work
and $\Ot(2^i)$ depth.
Eventually, since $H_1=V$, we obtain the all-pairs distances in $G$.
Summing through all inductive steps,
the work is $O(nm\log{n}\log{d})$, whereas the
depth is $\Ot\left(2^1+\ldots+2^K\right)=\Ot\left(2^K\right)=\Ot(d)$.

The last thing missing in the above algorithm is the induction base, i.e., computing
$\dist_G(s,t)$ for all $s,t\in H_d$
in only $\Ot(d)$ depth.\footnote{The all-pairs distances for $H_d$ can
be lifted to distances for all pairs in $(H_d\times V)\cup (V\times H_d)$
using $\Ot(nm)$ work and $\Ot(d)$ depth
as in the inductive step by setting $H_{2d}:=H_d$.}
This is where we depart from Bernstein's approach~\cite{Bernstein16}.
Note that lengths $\dist_G^{d+1}(s,t)$ of shortest at-most-$(d+1)$ hop
paths from all $s\in H_d$ can be computed using $(d+1)$-step
parallel Bellman-Ford using $O(|H_d|dm)=O(nm\log{n})$ work
and $\Ot(d)$ depth.
In order to combine these bounded-hop lengths into
actual distances between $s,t\in H_d$, we switch
to the repeated-squaring algorithm for dense graphs,
as was also done in the parallel SSSP~\cite{BringmannHK17} and transitive closure~\cite{UllmanY91} algorithms.\footnote{The same strategy of switching to a dense-graph algorithm
has also proved useful in obtaining a deterministic incremental algorithm for APSP in weighted directed graphs~\cite{KarczmarzL19}.}
Namely, we run this algorithm on a complete graph $G_d$ on $H_d$
whose edge $uv$ has weight $\dist_G^{d+1}(u,v)$.
The correctness of this approach follows by observing
that a shortest $s\to t$ path that uses more than $d+1$ hops
has to encounter an intermediate vertex $z\in H_d$ after no more than $d+1$ hops.
The $O(\polylog{n})$-depth repeated-squaring algorithm adds an $\Ot((n/d)^3)$ term to the
work but does not increase the depth, up to polylogarithmic factors.
\paragraph{Computing hub sets deterministically.} Now let us briefly describe how to derandomize the above
parallel APSP algorithm by replacing sampling with a deterministic preprocessing step
that computes all hub sets $H_1,\ldots,H_d$ within $\Ot(nm)$ work and $\Ot(d)$ depth.

Previous results in the parallel
and dynamic graph algorithms literature~\cite{BringmannHK17, KarczmarzL19, King99} that
used deterministic $h$-hub sets typically obtained them by (1) computing (or maintaining -- in the dynamic setting)
at-most-$h$-hop shortest paths between all pairs of vertices $s,t\in V$,
and (2) running a greedy $O(\log{n})$-approximation algorithm for
computing a $O((n/h)\log{n})$-size hitting set of a family of $h$-element
subsets of an $n$-element universe (see e.g.,~\cite{King99} for analysis).
Unfortunately, the first step of this approach seems to require $\Omega(nmh)$ time.

To obtain an improved bound, we reuse the inductive approach. Specifically,
we show that given a hub set $H_{2^i}$, one can obtain a $2^{i+1}$-hub set $H_{2^{i+1}}$
of size $O((n/2^i)\cdot \log{n})$
by running the aforementioned greedy hitting set algorithm on shortest
$2^i$-hop paths from $H_{2^i}$\footnote{A somewhat similar trick has been used in~\cite{KarczmarzL19}
for improving the state-of-the-art randomized partially dynamic APSP algorithms~\cite{Bernstein16, BaswanaHS07} from Monte Carlo to Las Vegas.}.
Constructing these paths using Bellman-Ford costs $O(|H_{2^i}|\cdot 2^i\cdot m)=O(nm\log{n})$ time
and requires $\Ot(2^i)$ depth.
Luckily, the deterministic greedy hitting set algorithm has its nearly work-efficient parallel
version with $O(\polylog{n})$ depth~\cite{BergerRS94}.
Therefore, one can construct $H_d$ using $\Ot(nm)$ work and $\Ot(d)$ depth.

It is worth noting that the analysis of the above algorithm breaks if the
shortest negative cycle in~$G$ has at most $d$ edges.
We stress that this is not a problem for the APSP application though, since any negative cycles
in~$G$ make the APSP problem infeasible.
That being said, our deterministic hub set algorithm can be easily extended to correctly
report the shortest negative cycle within the same bounds
and thus match the randomized bound of~\cite{OrlinSW18} for the shortest negative cycle problem.

We also remark that even though the hub sets are very useful in designing
dynamic graph algorithms, it is unlikely that our construction can help
match the best known randomized bounds for dynamic problems (like~\cite{Bernstein16})
using deterministic algorithms.
The power of a randomly sampled hub set $H$ (and, possibly, the so-called \emph{oblivious adversary} assumption)
lies in the fact that $H$ retains its guarantees through all, in particular \emph{future}, versions of a dynamic
graph.
\paragraph{Comparison to the transitive closure algorithm of~\cite{UllmanY91}.}
Ullman and Yannakakis~\cite{UllmanY91} use a single randomly sampled $d$-hub set $H_d$
of size $O((n/d)\log{n})$. In similar way they compute reachability between the
nodes of $H_d$ in $\Ot(nm+(n/d)^3)$ work and $\Ot(d)$ depth -- they first apply
limited-depth BFS from all $H_d$, and then use repeated squaring.
Adding the $\Ot((n/d)^2)$ ``shortcuts'' between $H_d$ to $G$ reduces
$G$'s diameter to $\Ot(d)$. Finally, a limited-depth BFS is run from each
source in parallel. Since the augmented graph has $\Ot(m+(n/d)^2)$ edges,
this takes $\Ot(nm+n^3/d^2)$ work and $\Ot(d)$ depth.
The efficiency of this approach crucially relies on the fact that BFS
has near-linear work and, as a result, does not
generalize to real-weighted graphs which require using Bellman-Ford.
\paragraph{Planar graph applications.} All of the numerous consequences
of our parallel APSP algorithm for planar graphs essentially follow by
showing improved parallel and sequential bounds for the following problem,
which we call \emph{dense distance graph APSP} (DDG APSP).

Let $H_1,H_2\in \TG(G)$ be two pieces of a recursive decomposition of $G$
and let $b=|\bnd{H_1}|+|\bnd{H_2}|$.
We would like to compute all-pairs shortest paths
in the graph $\DDG_{H_1}\cup \DDG_{H_2}$.

A parallel algorithm solving the above problem using $T(b,d)=\Omega(b^2)$ work and $\Ot(d)$ depth,
plugged in the framework of Fakcharoenphol~and~Rao~\cite{FR}, implies:
\begin{itemize}
  \item An $\Ot\left(n+T(\sqrt{n},d)\right)$-work and $\Ot(d)$-depth algorithm
    for negative-cycle detection on real-weighted planar graphs.
    Via known reductions~\cite{MillerN95, ChechikN20}, the same bound
    can be achieved for the feasible flow, bipartite perfect matching, and
    replacement paths problems.
  \item $\Ot\left(nd+T(\sqrt{n},d)\right)$-time sequential strongly polynomial algorithms
    for minimum mean and minimum cost-to-time ratio cycle problems.
\end{itemize}
Moreover, a sequential algorithm solving the DDG APSP problem in $S(b)=\Omega(b^2)$ time implies
an $O(S(\sqrt{n})\log{n})$ algorithm for computing all external dense distance graphs
$\DDG_{G-H}$ for $H\in \TG(G)$.

Fakcharoenphol and Rao's~\cite{FR} algorithm for solving DDP APSP
uses Johnson's~\cite{Johnson77} approach: first a feasible price function is computed using
Bellman-Ford algorithm to reduce the task to the non-negatively weighted case. Subsequently, Dijkstra's algorithm is run
from each of $\Omega(b)$ sources. Fakcharoenphol and Rao gave a very efficient
Dijkstra implementation running in $O(b\log^2{b})$~time~on $\DDG_{H_1}\cup\DDG_{H_2}$.
They also showed that a single step of Bellman-Ford can be performed
on $\DDG_{H_1}\cup \DDG_{H_2}$ in $O(b\log{b})$ time.
Klein et al.~\cite{KleinMW10} later noticed that a Bellman-Ford
step can be performed in $O(b\cdot \alpha(b))$ time.
By these bounds, the dense distance graph APSP can be solved
sequentially in $O(b^2\log{b})$ time and using a parallel
algorithm with $\Ot(b^2)$ work and $\Ot(b)$ depth.

We show that the dense distance graph APSP problem can be solved
using a parallel algorithm with work $\Ot(b^2+(b/d)^3)$ and depth $\Ot(d)$.
This essentially follows by using a more efficient parallel implementation
of a Bellman-Ford step in the algorithm of Theorem~\ref{t:apsp}.
As observed in~\cite{FR}, a Bellman-Ford step on a dense distance graph
can be rephrased as computing column minima of a certain matrix with Monge property (a \emph{Monge matrix}, in short).
Since column minima of a Monge matrix can
be computed using a polylogarithmic time parallel algorithm~\cite{AggarwalKPS97},
a single Bellman-Ford step can be implemented
using only $\Ot(b)$ work and $O(\polylog{n})$ depth, even though $\DDG_{H_1}\cup \DDG_{H_2}$ has $\Omega(b^2)$ edges.

Moreover, we show a faster sequential algorithm for DDG APSP with
$O(b^2\log{b}\cdot \log{\log{b}}\cdot \alpha(b))$ running time.
First, we observe that our parallel APSP algorithm (given a feasible price function) has
a sequential implementation with $O\left(n\cdot \mathcal{B}(n,m)\cdot \log{n}\cdot \log{d}+\left(\frac{n}{d}\log{n}\right)\cdot \mathcal{D}(n,m)\right)$
time, where $\mathcal{B}(n,m)$ denotes the cost of a single Bellman-Ford step,
and $\mathcal{D}(n,m)$ denotes the cost of a single Dijkstra step
(on a graph with $n$ vertices and $m$ edges).
Second, we leverage the asymmetry between the $O(b\cdot \alpha(b))$ cost of a Bellman-Ford
step and the $O(b\log^2{b})$ cost of running FR-Dijkstra on $\DDG_{H_1}\cup \DDG_{H_2}$.
To obtain our improved bound, it is enough to choose $d=\log^2{b}$.

\section{Preliminaries}\label{sec:preliminaries}
In this paper we deal with \emph{real-weighted directed} graphs.
We write $V(G)$ and $E(G)$ to denote the sets of vertices and edges of $G$, respectively.
A graph $H$ is a \emph{subgraph} of $G$, which~we~denote by $H\subseteq G$, if
and only if $V(H)\subseteq V(G)$ and $E(H)\subseteq E(G)$.
We write $uv\in E(G)$ when referring to edges of $G$ and use $\wei_G(uv)$
to denote the weight of $uv$.
If $uv\notin E$, we assume $\wei_G(uv)=\infty$.

A sequence of edges $P=e_1\ldots e_k$, where $k\geq 1$ and $e_i=u_iv_i\in E(G)$, is called
an $s\to t$ path in~$G$ if $s=u_1$, $v_k=t$ and $v_{i-1}=u_i$ for each $i=2,\ldots,k$.
For brevity we sometimes also express $P$ as a sequence of $k+1$ vertices $u_1u_2\ldots u_kv_k$ or as a subgraph of $G$ with vertices $\{u_1,\ldots,u_k,v_k\}$
and edges $\{e_1,\ldots,e_k\}$.
The \emph{hop-length} $\hops{P}$ is equal to the number of edges in $P$.
We also say that $P$ is a \emph{$k$-hop path}.
The \emph{length} of the path $\len(P)$ is defined as $\len(P)=\sum_{i=1}^k\wei_G(e_i)$.
For convenience, we sometimes consider a single edge $uv$ as a path of hop-length $1$.
If $P_1$ is a $u \to v$ path and $P_2$ is a $v \to w$ path, we denote by $P_1\cdot P_2$ (or simply $P_1P_2$) a path obtained by concatenating $P_1$ with $P_2$.

We define $\disth{k}_G(u,v)$ to be the length of the shortest path from $u$ to $v$ among paths of at most~$k$ edges.
Formally, $\disth{k}_G(u,v)=\min\{\len(P):u\to v=P\subseteq G\text{ and }\hops{P}\leq k\}$.
The \emph{distance} $\dist_G(u,v)$ between the vertices $u,v\in V(G)$ is the length
of the shortest $u\to v$ path in $G$, or $\infty$, if no $u\to v$ path
exists in $G$. In other words, $\dist_{G}(u,v)=\min_{k\geq 0}\dist_G^k(u,v)$.
Note that the distance is well-defined only if $G$ contains no negative cycles.
It is well known that $G$ has no negative cycles if and only if
there exists a \emph{feasible price function} $p:v\to\mathbb{R}$ satisfying
$\wei_G(e)+p(u)-p(v)\ge 0 $ for all $uv=e\in E(G)$. We define \emph{minimal} paths
as follows.
\begin{definition}
We call a $u\to v$ path $P\subseteq G$ \emph{minimal} if $\len(P)=\dist^{|P|}_G(u,v)<\dist^{|P|-1}_G(u,v)$.
\end{definition}
\begin{observation}\label{o:minimal}
  All subpaths of a minimal path are also minimal.
\end{observation}

If the graph $G$ that we refer to is clear from the context, we sometimes omit the subscript $G$ and write $\wei(uv)$, $\dist(u,v)$, $\disth{k}(u,v)$ etc.
instead of $\wei_G(u,v)$, $\dist_G(u,v)$, $\disth{k}_G(u,v)$, etc., respectively.
\paragraph{Parallel model.}
Formally, we use the work-depth model as used in recent literature on parallel reachability and shortest paths, e.g.,~\cite{Fineman18, Li20}.
An algorithm in this PRAM model differs from a sequential algorithm
only by the inclusion of the \emph{parallel foreach} loops.
The work of an algorithm is the total number of operations performed
if all parallel foreach loops were executed sequentially.
The \emph{depth} of an algorithm is the total number of operations
performed by sequential steps, plus the sum, over all parallel foreach loops,
of the maximum (sequential) iteration cost of a loop.
In order to not focus on low-level details,
we specify neither what is the depth overhead of a $k$-way parallel foreach loop,
nor what is the precise shared memory access model (e.g., EREW, CREW).
Instead, we state all our parallel work and depth bounds using $\Ot(\cdot)$ notation
that suppresses polylogarithmic
factors. This is justified by the existence of general reductions between these different PRAM variants
that yield only polylogarithmic multiplicative overhead in work and depth~\cite{JaJa92}.
\paragraph{Bellman-Ford algorithm.}
The Bellman-Ford algorithm is a classical algorithm for computing
shortest paths from a single source $s$ or detecting a negative cycle if one exists.
It maintains a distance label vector $d:V\to\mathbb{R}$,
where initially $d(s)=0$ and $d(v)=\infty$ for all $v\in V\setminus\{s\}$,
and proceeds in \emph{steps}.
Classically, a \emph{Bellman-Ford step} consists of performing \emph{edge relaxations}, i.e., substitutions $d(v):=\min(d(v),d(u)+\wei(uv))$ for all edges $e=uv$, in any order.
It is well-known that: (1) if $d(v)\neq \infty$ then $d(v)$ is a length of some $s\to v$ path in $G$, (2) after~$k$ Bellman-Ford steps
we have $d(v)\leq \dist_G^k(s,v)$ for all $v\in V$,
(3) if $d(u)+\wei(uv)<d(v)$ for some $uv\in E$ after $n-1$ steps,
then~$G$ has a negative cycle,
(4) if $G$ has no negative cycle and $s$ can reach every other vertex, then the obtained
distance label vector constitutes a feasible price function of $G$.
A single Bellman-Ford step clearly takes $O(m)$ time, so the
Bellman-Ford algorithm runs in $O(mk)$ time if $k$ steps are performed.
However, in general, the results of individual relaxations in a single step
may depend on the results of relaxations that happened earlier in that step.

In this paper we actually use a variant of the Bellman-Ford algorithm that might converge
to the answer slower, but is easier to reason about.
Namely, in a single \emph{step}, \emph{for all $v$ at once}, we replace $d(v)$ with
$\min\left(d(v),\min_{uv\in E}\{d(u)+\wei(uv)\}\right)$. Equivalently,
at the beginning of a step we could store a copy of vector $d$ as $d'$,
and then again \emph{relax} the edges in any order,
where the relaxation is now a substitution $d(v):=\min(d(v),d'(u)+\wei(uv))$.
It is easy to prove that in this variant,
after $k$ steps we have precisely $d(v)=\dist_G^k(v)$ for all $v\in V$.
The properties (3) and (4) of the classical Bellman-Ford algorithm hold for this variant as well.
Moreover, the result of each relaxation now only depends
on the distance labels in the previous step.
Consequently, for this variant, a single Bellman-Ford step can be clearly performed
in parallel using $\Ot(m)$ work and $O(\polylog{n})$ depth.

\paragraph{Hitting sets.} Let $\mathcal{F}$ be a collection of subsets of some universe $U$.
Then $X\subseteq U$ is called a \emph{hitting set} of $\mathcal{F}$ if $X\cap S\neq \emptyset$
for all $S\in \mathcal{F}$.

\begin{lemma}[\cite{BergerRS94, King99}]\label{l:hitting}
  Let $\Pi$ be a collection of $k$ \emph{simple} $h$-hop paths of $G$ (i.e., $k$ $(h+1)$-element subsets of~$V(G)$). A hitting set $\Pi$ of
  size $O((n/h)\cdot \log{k})$ can be computed in a deterministic way:
  \begin{itemize}
    \item sequentially using a greedy algorithm in $O(kh+n)$ time,
    \item using a parallel algorithm with $O(\polylog{n})$ depth and $\Ot(kh+n)$ work.\footnote{The algorithm of Berger et al.~\cite{BergerRS94} actually
      produces a hitting set a constant-factor larger than the greedy algorithm.}
  \end{itemize}
\end{lemma}
\section{A Parallel APSP Algorithm}\label{s:apsp}
In this section we describe the main result of this paper. Our algorithm
will use the concept of an~$h$-hub set, as defined below.

\begin{definition}\label{d:hubs}
  We call a set $H_h\subseteq V$ a $h$-hub set if for all $u,v\in V$ such
  that $\dist^h_G(u,v)<\dist^{h-1}_G(u,v)$ there exists a minimal path $P=u\to v$ in $G$
  such that $|P|=h$ and $P$ goes through a vertex of $H_h$.
\end{definition}
The following randomized construction of $h$-hub sets is attributed
to Ullman and Yannakakis~\cite{UllmanY91}.
\begin{fact}[\cite{UllmanY91}]\label{f:sample}
  For any $h\in [1,n]$, a random $\Theta\left(\frac{n}{h}\log{n}\right)$-element subset of $V$
  constitutes an $h$-hub set of $G$ with high probability\footnote{Here we abuse the notation slightly.
  Formally, by increasing the constant $c\geq 1$ hidden in the $\Theta$ notation, one can achieve
  probability $1-1/n^c$.}.
\end{fact}

First, we will show how having hub sets for various values $h$ leads to a randomized algorithm,
whereas the deterministic construction of hubs will follow.
\subsection{The Algorithm}
In the remaining part of this section we assume that $G$ has no negative cycles.
We will deal with negative cycle detection later on.

Let $d\in [1,n]$ be a parameter that controls the depth of our algorithm.
By possibly decreasing the demanded depth by just a constant factor, we can
assume, without loss of generality, that $d=2^K$, where $K$ is an integer.
We first show how to compute APSP given hub sets \linebreak $V=H_1,H_2,\ldots,H_{2^i},\ldots H_{d}$,
where each $H_{2^i}$ has size $O((n/2^i)\cdot \log{n})$.
By Fact~\ref{f:sample}, all these hub sets can be obtained with $\Ot(n)$ work and $O(\polylog{n})$ depth using sampling.

The first step of our algorithm is to compute shortest $\leq (d+1)$-hop distances $\dist^{d+1}_G(s,t)$
for all $s,t\in H_d$. This can be done by running $(d+1)$ steps of the Bellman-Ford algorithm from all $s\in H_d$ in parallel
using $\Ot(|H_d|\cdot d\cdot m)=\Ot(nm)$ work and $\Ot(d)$ depth,
or in $O(nm\log{n})$ time sequentially.

Let $G_d$ be defined as a complete graph on $H_{d}$
with weights given by $\wei_{G_d}(uv)=\dist^{d+1}_G(u,v)$ for all $u,v\in H_d$.
The second step is computing the actual shortest paths $\dist_G(s,t)$ for all $s,t\in H_{d}$
by running the APSP algorithm based on repeated-squaring (of the weighted adjacency
matrix using min-plus product) on the graph $G_d$.
\begin{lemma}
  Let $s,t\in H_{d}$. Then $\dist_{G_d}(s,t)=\dist_G(s,t)$.
\end{lemma}
\begin{proof}
  Since the edge lengths in $G_d$ encode path lengths in $G$, we clearly have $\dist_{G_d}(s,t)\geq \dist_G(s,t)$.
  Let $h$ be the hop-length of a minimal shortest $s\to t$ path in $G$, i.e., minimum $h$ such that $\dist^h_G(s,t)=\dist_G(s,t)$.
  We prove $\dist_{G_d}(s,t)\leq \dist_G(s,t)$ by induction on $h$.

  If $h\leq d+1$, then $\dist_G(s,t)=\dist^{d+1}_G(s,t)=\wei_{G_d}(st)\geq \dist_{G_d}(s,t)$.
  Otherwise, suppose $h>d+1$ and consider some minimal shortest $s\to t$ path $P$ in $G$.
  Let us write $P=P'QP''$, where $|P'|,|P''|\geq 1$, $|Q|=d$, and $Q$ is a $p\to q$ path, $p,q\in V$.
  Since $P$ is minimal, so is $Q$ (by Observation~\ref{o:minimal}) and, as a result, we have $\dist^{d-1}_G(p,q)>\dist^d_G(p,q)$.
  Therefore, by the definition of a $d$-hub set, there exists a path $Q'=p\to q$ in $G$ with length $\dist^d_G(p,q)=\len(Q)$,
  $|Q'|=|Q|$, and going through a vertex $z\in H_d$.
  Consequently, $P'Q'P''$ is a minimal shortest $s\to t$ path with $z\in H_d$ as an intermediate vertex.
  By Observation~\ref{o:minimal}, this implies $\dist^{h-1}_G(s,z)=\dist_G(s,z)$ and $\dist^{h-1}_G(z,t)=\dist_G(z,t)$.
  By the inductive assumption we conclude $\dist_{G_d}(s,z)=\dist_G(s,z)$ and $\dist_{G_d}(z,t)=\dist_G(z,t)$.
  Finally, by $z\in V(G_d)$ and triangle inequality we obtain $\dist_{G_d}(s,t)\leq \dist_{G_d}(s,z)+\dist_{G_d}(z,t)=\dist_G(s,z)+\dist_G(z,t)=\dist_G(s,t)$.
\end{proof}

The repeated squaring APSP algorithm has $\Ot(|H_d|^3)=\Ot((n/d)^3)$ work
and $O(\polylog{n})$ depth (one can also think of the min-plus product as
$n$ Bellman-Ford steps that can be performed in parallel, and hence a single product requires $O(\polylog{n})$ depth).
If one wanted to implement this step sequentially, the Floyd-Warshall
APSP algorithm would yield $O((n/d)^3\log^3{n})$ time.

Finally, the last step is to inductively compute, for each $k=K,\ldots,$ down to $0$, the distances $\dist_G(s,t)$ for all
pairs $(s,t)\in (H_{2^k}\times V)\cup (V\times H_{2^k})$.
Recall that we have set $H_0=V$, so after completing the step for $k=0$, this will give the all-pairs distance matrix
as desired.

For convenience, let us set $H_{2^{K+1}}:=H_{2^K}$.
Let us focus on computing $\dist_G(s,t)$ for $s\in H_{2^k}$ and all $t\in V$
assuming that the steps for larger $k$ have already been completed
and so the distances $\dist_G(u,v)$ for $(u,v)\in (H_{2^{k+1}}\times V)\cup (V\times H_{2^{k+1}})$
are already known.
Actually, for the inductive step we only require these distances for the pairs from $(H_{2^{k+1}}\times H_{2^k})\cup (H_{2^k}\times H_{2^{k+1}})$,
which implies that in the first step (for $k=K$) they can be retrieved from the distance matrix of $G_d$.
Computing $\dist_G(s,t)$ for $(s,t)\in V\times H_{2^k}$
is symmetric.
Let $G_{k,s}$ be $G$ with an auxiliary edge $e_v=sv$ of weight $\wei_{G_{k,s}}(e_v)=\dist_G(s,v)$
added for all $v\in H_{2^{k+1}}$.
Observe that all the auxiliary edges' weights have been already computed in
the previous step.
We compute the desired distances $\dist_G(s,t)$ by running $2^{k+1}+1$ steps of the Bellman-Ford algorithm
from each $s\in H_d$ (in parallel).
This costs $\Ot(|H_{2^k}|\cdot m\cdot 2^k)=\Ot(mn)$ work and $\Ot(2^k)$ depth
in parallel, or $O(nm\log{n})$ time sequentially.

Note that through all $k$, the total work is $\Ot(nm)$, whereas the depth is $\Ot(d)$.
The sequential time cost of the final phase is $O(nm\log{n}\log{d})$.
The correctness of the above algorithm follows from the lemma below.
\begin{lemma}
  For any $(s,t)\in H_{2^k}\times V$, $\dist_{G_{k,s}}^{2^{k+1}+1}(s,t)=\dist_G(s,t)$.
\end{lemma}
\begin{proof}
  Let $b:=2^{k+1}+1$.
  Since $G\subseteq G_{k,s}$, and the auxiliary edges encode some path lengths in $G$, we clearly have
  $\dist_{G_{k,s}}^{b}(s,t)\geq \dist_{G_{k,s}}(s,t)=\dist_G(s,t)$.

  Let us now prove $\dist_{G_{k,s}}^{b}(s,t)\leq \dist_G(s,t)$.
  Let $P$ be a minimal shortest $s\to t$ path in $G$.
  If $|P|\leq b$, then our claim holds by $G\subseteq G_{k,s}$. 
  Suppose $|P|>b$.
  Let us write $P=QR$, where $|R|=2^{k+1}$ and $R=u\to t$.
  Observe that the path $R$ is shortest and minimal.
  Hence, by the definition of $H_{2^{k+1}}$, there exists
  a minimal shortest path $R'=u\to t$ going through a vertex $z\in H_{2^{k+1}}$
  and $|R'|=2^{k+1}$.
  Therefore, $QR'$ is also a minimal shortest $s\to t$ path.
  Let us express $QR'$ as $P_1P_2$, where $P_2$ is a minimal shortest $z\to t$ path.
  Note that we have $|P_2|\leq 2^{k+1}$.
  Moreover, $\len(P_1)=\dist_G(s,z)=\wei_{G_{k,s}}(e_z)$.
  Hence, the $s\to t$ path $e_zP_2$ of length $\dist_G(s,t)$ is contained in $G_{k,s}$ and its hop-length
  does not exceed $2^{k+1}+1=b$.
  Its existence implies $\dist_{G_{k,s}}^b(s,t)\leq \dist_G(s,t)$.
\end{proof}
Hence, we obtain the following lemma.
\begin{lemma}\label{l:alg}
  Let $d=2^K\leq n$ and assume that $G$ does not contain a negative cycle. Given
  a collection of sets $H_{2^k}$, where $k=0,\ldots,K$,
  such that $H_{2^k}$ is a $2^k$-hub set of $G$ of size $O((n/2^k)\log{n})$,
  all-pairs shortest paths can be computed in parallel using $\Ot(nm)$ work and $\Ot(d)$ depth,
  or sequentially in $O(nm\log{n}\log{d})$ time.
  \end{lemma}

\subsection{Deterministic Construction of Hubs}
In this section we show how to construct hub sets in a deterministic way. We start
with a few technical lemmas.
\begin{lemma}\label{l:rec}
  Let $H_h$ be a $h$-hub set of $G$. Let $\Pi$ be a collection
  of minimal $h$-hop paths $P_{st}=s\to t$, one path for
  each pair $(s,t)\in H_h\times V$ for which such a minimal path exists.

  Let $B$ be a hitting set of $\Pi$.
  Then $B$ is a $2h$-hub set of $G$.
\end{lemma}
\begin{proof}
  We need to prove that, for each $u,v\in V$ such that $\dist^{2h}_G(u,v)<\dist^{2h-1}_G(u,v)$,
  there exists a minimal $2h$-hop $u\to v$ path in $G$ going through a vertex of $B$.
  To this end, let $P$ be some minimal $u\to v$ path such that $|P|=2h$.
    Split $P$ evenly into two paths $P_1P_2$ of hop-length $h$.
    Note that, by Observation~\ref{o:minimal}, every subpath of $P$, in particular $P_1=u\to w$, is minimal.
    Since $H_h$ is a hub set, there is a vertex $z\in H_h$ on some minimal
    path $Q=u\to w$ such that $|Q|=h$.
    Let $P'=QP_2$, and let us express $P'=RST$ so that $S=z\to y$ is a path satisfying $|S|=h$.
    Note that by minimality of~$P'$, $S$ is also minimal.
    Therefore, $\dist_G^h(z,y)<\dist_G^{h-1}(z,y)$, and thus we have a minimal $z\to y$ path~$P_{zy}$
    of hop-length~$h$
    in $\Pi$.
    Finally, observe that $RP_{zy}T$ is a minimal $u\to v$ path of hop-length~$2h$
    and goes through a vertex of~$B$ by the definition of a hitting set of $\Pi$.
\end{proof}

\begin{lemma}\label{l:simple}
  Suppose that $G$ has no negative cycles with at most $h$ edges.
  Then every minimal $h$-hop path in $G$ is simple.
\end{lemma}
\begin{proof}
  A non-simple path contains a cycle.
  If a minimal $h$-hop path $P$ contained a non-negative cycle, there would
  exist a path of hop-length $<h$ and length no more than $\ell(P)$,
  thus contradicting minimality of $P$.
  If $P$ contained a negative cycle, the cycle would have at most $h$ edges.
\end{proof}

\begin{lemma}\label{l:comp}
  Suppose that $G$ has no negative cycles with at most $h$ edges.
  Let $H_h$ be a $h$-hub set of~$G$. Then in $O(|H_h|hm)$ time we
  can:
  \begin{itemize}
    \item construct a $2h$-hub set $H_{2h}$ of $G$ such that $|H_{2h}|=O((n/h)\cdot\log{n})$,
    \item find a shortest negative cycle in $G$ with no more than $2h$ edges, if one exists.
  \end{itemize}
\end{lemma}

\begin{proof}
  To find $H_{2h}$, we first compute shortest $\leq h$-hop
  paths from all $s\in H_h$ to all $v\in V$ using~$h$ steps of Bellman-Ford algorithm
  in $O(|H_h|hm)$ time.
  Note that indeed minimal $h$-hop paths can be inferred from the Bellman-Ford
  execution by storing (1) the distance labels $d_i$ from all the steps,
  and (2) the predecessor vectors $p_i$ such that
  $$d_{i+1}(v)=\min_{uv\in E}\{d_i(u)+\wei(uv)\}=d_i(p_i(v))+\wei(p_i(v)v).$$
  Clearly, an $s\to v$ minimal path of hop-length $h$ exists if $d_h(v)<d_{h-1}(v)$
  and it can be easily reconstructed from the predecessor vectors.
  Since no negative cycle in $G$ has $\leq h$ edges, by Lemma~\ref{l:simple}, the computed
  minimal $h$-hop paths are all simple.
  As a result, by Lemma~\ref{l:hitting}, we can compute a hitting set $B$ of size $O((n/h) \log{n})$ of these paths
  in $O((|H_h|\cdot n)\cdot h)$ time.
  By Lemma~\ref{l:rec}, $B$ constitutes a $2h$-hub set of $G$.

  Suppose a shortest negative cycle $C=v_1\ldots v_k$ with $|C|=k\in (h,2h]$ exists in $G$.
  Since $C$ has a minimal number of hops (in particular, it is simple), the subpath $P=v_1\to v_{h+1}$ of $C$
  satisfies $\dist_G^h(v_1,v_{h+1})\leq \ell(P)<\dist_G^{h-1}(v_1,v_{h+1})$.
  As a result, and by the definition of $H_{h}$, there exists a $h$-hop minimal path
  $P'=v_1\to v_{h+1}$ such that $z\in H_h\cap V(P')$ and $\ell(P')\leq \ell(P)$.
  Hence, if we replace the prefix $P$ of $C$ with $P'$, we obtain
  another shortest negative cycle $C'$ with $|C'|\in (h,2h]$ and going through $z$.
  It follows that $\dist_G^{2h}(z,z)<0$.
  By performing $2h$ Bellman-Ford steps from all $z\in H_h$ independently
  we can thus find the smallest $k\leq 2h$ such that $\dist^k(z,z)<0$ for some $z\in H_h$,
  if one exists.
  This $k$ clearly equals $|C|$.
  Since we have to check all $z\in H_h$, this takes $O(|H_h|hm)$ time through all $z$.
  \end{proof}

We are now ready to describe the preprocessing step of our APSP algorithm
that computes the hub sets $H_1,\ldots,H_{d=2^K}$ and possibly detects the shortest negative
cycle.
Let us first discuss a sequential algorithm.
We proceed inductively, using Lemma~\ref{l:comp}.

We set $H_{2^0}=V$. 
For each $k=0,\ldots,K-1$ we proceed as follows.
We maintain an invariant that $G$ contains no negative cycles of hop-length
at most $2^k$ and $|H_k|=O((n/2^k)\log{n})$.
Hence, the invariant is true initially for $k=0$ since,
clearly, there are no single-edge negative cycles.
By our invariant, we can apply Lemma~\ref{l:comp}
for $h=2^k$ and thus in $O(nm\log{n})$ time either detect a shortest negative
cycle with no more than $2h$ edges (and declare the APSP problem infeasible)
or construct a $2^{k+1}$-hub set $H_{2^{k+1}}$ of size
$O((n/2^{k+1})\log{n})$ and guarantee the invariant for $k+1$.

The above sequential algorithm trivially implies a parallel algorithm:
all the above $O(\log{d})$ steps of our computation can be implemented using a number of parallel invocations of an $O(d)$-step Bellman-Ford algorithm
and an $O(\polylog{n})$-depth
parallel computation (as given in Lemma~\ref{l:hitting}) of a small hitting set of a collection
of simple paths with a total hop-length of $O(n^2\log{n})$ (the collection
in step $k$ has $O((n^2/2^k)\cdot \log{n})$ paths of length $\Theta(2^k)$).
We thus obtain the following theorem.

\begin{theorem}\label{t:hubs}
  Let $d=2^K\leq n$. Then, a shortest negative cycle of $G$, provided it has
  hop-length at most $d$, can be computed deterministically:
  \begin{itemize}
    \item in $O(nm\log{n}\log{d})$ time using a sequential algorithm,
    \item using a parallel algorithm with $\Ot(nm)$ work
      and $\Ot(d)$ depth.
  \end{itemize}
  If no negative cycle of hop-length at most $d$ exists in $G$, then the algorithm
  can produce, within the same time bounds, a collection of sets $H_{2^k}$, where $k=0,\ldots,K$,
  such that $H_{2^k}$ is a $2^k$-hub set of $G$ of size $O((n/2^k)\cdot \log{n})$.
\end{theorem}
Together with Lemma~\ref{l:alg} this finishes the proof of our main theorem.
\tapsp*
As a corollary, we also obtain the following theorem.
\tcycle*

\section{Planar Graph Algorithms}
In this section we present our improved planar graph algorithms. We start
from describing the framework used by Fakcharoenphol and Rao~\cite{FR}
to solve the negative cycle detection problem on planar graphs in near-linear time.
This framework forms the basis for all our parallel and sequential algorithms.
We remark that the negative cycle detection algorithm of~\cite{FR} has been subsequently
improved by Klein et al.~\cite{KleinMW10}, whereas the currently best known bound $O(n\log^2{n}/\log\log{n})$ is due to Mozes and Wulff-Nilsen~\cite{MozesW10}.
These algorithms take a slightly simpler recursive approach, but
the fundamental difference compared to~\cite{FR} is using Klein's MSSP algorithm~\cite{MSSP} on each piece $H$ with reduced costs for computing $\DDG_H$
in $O((|V(H)|+|\bnd{H}|^2)\log{n})$ time.
This approach allows to avoid using the costly FR-Dijkstra (Lemma~\ref{l:fr}) and
thus saves at least a $O(\log{n})$ factor in the running time.
However, it seems that Klein's MSSP algorithm is inherently sequential
and this is why in the following we will stick to the original approach of~\cite{FR}.

Transferring the approach of~\cite{FR} to the parallel setting directly leads to
$\Ot(\sqrt{n})$-depth work-efficient parallel algorithms for single source shortest paths and negative cycle detection. 
We first show how to improve the depth to $\Ot(n^{1/6})$ using our parallel APSP algorithm.
Using this, we later show improved algorithms for minimum cost-to-time ratio cycle problem
and external dense distance graphs computation.

\subsection{The Framework of Fakcharoenphol and Rao~\cite{FR}}
Let $H$ be a weighted plane digraph with a distinguished set $\bnd{H}$ of \emph{boundary vertices}
that necessarily lie on a $O(1)$ faces of~$H$.
We denote by $\DDG_H$ (a \emph{dense distance graph}) the complete
weighted graph on $\bnd{H}$ whose edge weights represent distances
between all pairs of vertices of $\bnd{H}$ in~$H$.
Fakcharoenphol and Rao~\cite{FR} developed the concept of a dense distance graph
along with efficient algorithms for constructing and processing DDGs
as a way to obtain their breakthrough $O(n\log^3{n})$-time algorithm for negative cycle detection
in a real-weighted planar digraph.
We now review a variant of this algorithm using slightly more modern terminology.

After augmenting $G$ to be connected and triangulated, $G$ is recursively
decomposed using small cycle separators~\cite{Miller86} of size $O(\sqrt{n})$
until the obtained pieces have constant size.
The decomposition procedure produces in $O(n\log{n})$ time a binary tree
$\TG(G)$ whose nodes correspond to subgraphs of~$G$ (\emph{pieces}), with the root being all of~$G$ and the leaves being pieces of constant size.
We identify each piece $H$ with the node representing it in $\TG(G)$. We can thus abuse notation and write $H\in \TG(G)$.
The \emph{boundary vertices} $\bnd{H}$ of a piece $H$ are vertices that $H$ shares
with some other piece $Q\in \TG(G)$ that is not $H$'s ancestor.
For convenience we extend the boundary set $\bnd{L}$ of a leaf piece $L$
to its entire vertex set $V(L)$.
It is known that (see e.g.,~\cite{BorradaileSW15, DBLP:conf/stoc/KleinMS13}) one can additionally guarantee that for each piece $H\in\TG(G)$, (1) $H$ is connected,
(2) $\bnd{H}$ lies on some $O(1)$ faces of $H$, and (3) $|\bnd{H}|=O(\sqrt{n})$. 
Moreover, one can assume that $\sum_{H\in\TG(G)}|V(H)|=O(n\log{n})$ and
$\sum_{H\in\TG(G)}|\bnd{H}|^2=O(n\log{n})$.

Given the decomposition, the algorithm processes the pieces $H\in\TG(G)$
bottom-up. Clearly, if any piece contains a negative cycle, the whole $G$ does so as well.
On the other hand, if $H$ contains no negative cycle, the dense distance graph
$\DDG_H$ on $\bnd{H}$ is well-defined.
Therefore, the computation for a piece $H$ either detects a negative cycle
in $H$ or produces $\DDG_H$ otherwise.
Let $H_1,H_2$ be the children of the node $H$ in $\TG(G)$ and
suppose neither of them contains a negative cycle.
Let $H'=\DDG_{H_1}\cup\DDG_{H_2}$.
It can be easily shown that (1) $H$ contains a negative cycle if and only if
$H'$ contains a negative cycle,
(2) for any $u,v\in \bnd{H}$, if $H$ contains no negative cycle, then $\dist_H(u,v)=\dist_{H'}(u,v)$.
Consequently, in the algorithm of~\cite{FR}, for each piece $H$ we first run Bellman-Ford-based SSSP algorithm on $H'$ which either
detects a negative cycle or produces a feasible price function $p$ on $H'$.
The second step is to compute all-pairs shortest paths on $H'$ by running $|\bnd{H_1}\cup\bnd{H_2}|$
Dijkstra-based single-source computations with edge costs in $H'$ reduced with $p$.
The graph $\DDG_H$ can be easily obtained from the computed distance
matrix since $\bnd{H}\subseteq \bnd{H_1}\cup \bnd{H_2}$.

Since each $\DDG_H$ has $|\bnd{H}|^2$ edges, using Bellman-Ford and Dijkstra naively would lead to \linebreak $\Ot\left(\sum_{H\in \TG(G)}|\bnd{H}|^3\right)=\Ot(n^{3/2})$ running time.
The main contribution of~\cite{FR} lies in showing that the special structure
of a dense distance graph (that is, the distance matrix behind $\DDG_H$ consists
of two so-called \emph{staircase Monge matrices}) can be leveraged to speed up
naive implementations of these algorithms.
The original implementations of~\cite{FR} have been slightly improved, and currently we have the following bounds.

\begin{lemma}[\cite{FR, KleinMW10}]\label{l:staircase}
  A single step of Bellman-Ford algorithm can be simulated on
  $\DDG_{H_1}\cup \DDG_{H_2}$ in $O\left((|\bnd{H_1}|+|\bnd{H_2}|)\alpha(n)\right)$
  time, where $\alpha(n)$ is the inverse Ackermann function.
\end{lemma}
The above lemma is a simple application of the algorithm of Klawe~and~Kleitman~\cite{KlaweK90} for
computing row minima of a staircase $m\times m$ Monge matrix in $O(m\alpha(m))$
time.

\begin{lemma}[FR-Dijkstra~\cite{FR, GawrychowskiK18}]\label{l:fr}
  Given a feasible price function $p$ on $\DDG_{H_1}\cup \DDG_{H_2}$, one can
  simulate Dijkstra's algorithm on $\DDG_{H_1}\cup \DDG_{H_2}$ in
  $O\left((|\bnd{H_1}|+|\bnd{H_2}|)\frac{\log^2{n}}{\log^2{\log{n}}}\right)$
  time.
\end{lemma}

The above lemmas imply that the running time of Fakcharoenphol~and~Rao's algorithm
is \linebreak
$O\left(\sum_{H\in\TG(G)}|\bnd{H}|\log^2{n}/\log^2\log{n}\right)=O(n\log^3{n}/\log^2\log{n})$.

\paragraph{A parallel implementation.} We note that the above
algorithm can be easily turned into a parallel algorithm
with $\Ot(n)$ work and $\Ot(\sqrt{n})$ depth, as explained below.

First, computing the decomposition within $\Ot(n)$ work
and $O(\polylog{n})$ depth follows from the fact that $O(\sqrt{n})$-size
simple cycle separators can be computed within these bounds~\cite{KleinS93, Miller86}.
Next, it is known that row minima of a staircase Monge matrix can
be computed using near-linear work and $O(\polylog{n})$
depth~\cite{AggarwalKPS97}.
This implies the following.

\begin{lemma}\label{l:staircase-parallel}
  A single step of Bellman-Ford algorithm can be simulated on
  $\DDG_{H_1}\cup \DDG_{H_2}$
  using  a parallel
  algorithm with $\Ot\left(|\bnd{H_1}|+|\bnd{H_2}|\right)$ work
  and $O(\polylog{n})$ depth.
\end{lemma}

Consequently, the Bellman-Ford algorithm can be
simulated on $\DDG_{H_1}\cup \DDG_{H_2}$ using
$\Ot\left((|\bnd{H_1}|+|\bnd{H_2}|)^2\right)$ work
and $\Ot(|\bnd{H_1}|+|\bnd{H_2}|)=\Ot(\sqrt{n})$ depth.

Finally, since each of $|\bnd{H}|$ Dijkstra's algorithm runs used to compute
$\DDG_H$ based on $\DDG_{H_1}$ and $\DDG_{H_2}$ are
independent, they can be performed in parallel to give
$\Ot(\sqrt{n})$ depth bound on this step.
As the depth of $\TG(G)$ is $O(\log{n})$, the
total work is $\Ot(n)$, whereas the  total depth is $\Ot(\sqrt{n})$.


\subsection{A Faster Nearly Work-Efficient Parallel SSSP Algorithm}\label{s:parallel-sssp}
In this section we apply our parallel APSP algorithm to obtain a polynomially
smaller bound on the depth required to detect a negative cycle
and to compute single-source shortest paths in a real-weighted
planar digraph.

First, observe that having an APSP algorithm that can handle
negative cycle detection on the fly allows us to replace
the two-phase Bellman-Ford-Dijkstra approach in the framework
of Fakcharoenphol and Rao.
Indeed, then for each piece $H$ all we do is computing $\DDG_H$
by running the all-pairs shortest paths algorithm directly on $\DDG_{H_1}\cup \DDG_{H_2}$,
without using price functions that reduce the
problem to the non-negatively weighted case.

Recall that the algorithm of Theorem~\ref{t:apsp} is a certain combination
of a number of limited-hop Bellman-Ford invocations on various graphs $G'$
obtained from $G$ by adding $O(n)$ auxiliary edges,
and a single run of a repeated squaring algorithm.
In fact, if we were able to execute a single Bellman-Ford step
on such a graph $G'$
using $O(t(n,m))$ work and $O(\polylog{n})$ depth, the
algorithm of Theorem~\ref{t:apsp} would have
$\Ot(n^2+n\cdot t(n,m)+(n/d)^3)$ work and $\Ot(d)$ depth
for any $d$.
We use this fact to prove the following lemma.

\begin{lemma}\label{l:apsp-piece-parallel}
  Let $H_1,H_2$ be the children of a node $H\in \TG(G)$.
  Let $b=|\bnd{H_1}|+|\bnd{H_2}|$.
  We can compute all-pairs shortest paths in $\DDG_{H_1}\cup\DDG_{H_2}$
  (and thus obtain $\DDG_H$), using a parallel algorithm
  with $\Ot(b^2)$ work and $\Ot(b^{1/3})$ depth.
  If the problem is infeasible, the algorithm
  detects a negative cycle within the same bounds.
\end{lemma}
\begin{proof}
  Let $X$ be the graph $\DDG_{H_1}\cup\DDG_{H_2}$ with
  some $O(b)$ auxiliary edges added.
  Every Bellman-Ford step in the algorithm of Theorem~\ref{t:apsp}
  run on $\DDG_{H_1}\cup \DDG_{H_2}$
  is performed on a graph of this form.
  Hence, it is enough to show how to perform a Bellman-Ford
  step on $X$, which can have $\Omega(b^2)$ edges.
  Since in a Bellman-Ford step the order of edge relaxations
  is arbitrary, we can first relax all edges
  of $\DDG_{H_1}\cup\DDG_{H_2}$ and then all the auxiliary edges.
  Relaxing the former takes $O(b)$ work and $O(\polylog{n})$
  depth by Lemma~\ref{l:staircase-parallel}.
  The latter can be relaxed within the same bounds
  naively.
  Hence, the APSP algorithm of Theorem~\ref{t:apsp} can be
  implemented so that it has $\Ot(b^2+(b/d)^3)$ work
  and $\Ot(d)$ depth for any $d$.
  By choosing $d=b^{1/3}$ we obtain the desired bounds.
\end{proof}

\tparallel*
\begin{proof}
To obtain a parallel negative cycle detection algorithm, we simply replace the Bellman-Ford and Dijkstra steps in the
computation for each piece $H\in\TG(G)$ in the
algorithm of Fakcharoenphol and Rao~\cite{FR} by
  a single computation of $\DDG_H$ from $\DDG_{H_1}$ and $\DDG_{H_2}$
  as in Lemma~\ref{l:apsp-piece-parallel}.
  The work remains $\Ot\left(\sum_{H\in \TG(G)}|\bnd{H}|^2\right)=\Ot(n)$,
  whereas the depth is $$\Ot\left(\max_{H\in\TG(G)}|\bnd{H}|^{1/3}\right)=\Ot((\sqrt{n})^{1/3})=\Ot(n^{1/6}).$$

  In order to solve the single-source shortest paths problem
  one can use a trick first described by Cohen~\cite{Cohen96} (and also used in~\cite{Karczmarz18}).
  Denote by $H^*$ a complete graph on $\bnd{H_1}\cup\bnd{H_2}$
  whose edge weights represent distances in $\DDG_{H_1}\cup \DDG_{H_2}$.
  Note that $H^*$ is precisely the graph that we compute using our new APSP
  algorithm to obtain $\DDG_H\subseteq H^*$.
  Now, consider the graph
  $$G^*=\left(\bigcup_{\text{non-leaf }H\in \TG(G)} H^*\right)\cup \left(\bigcup_{\text{leaf }L\in \TG(G)} L\right).$$
  Since $G\subseteq G^*$ and the edges in $G^*$ correspond to paths
  in $G$, it is clear that $\dist_G(s,t)=\dist_{G^*}(s,t)$ for all $s,t\in V(G)$.
  However, as proven in~\cite{Cohen96, Karczmarz18}, the graph $G^*$, although no longer planar, has hop-diameter $O(\log{n})$.
  In other words,
  we in fact have $\dist_{G^*}^{O(\log{n})}(s,t)=\dist_G(s,t)$ for all $s,t\in V$.
  As a result, shortest paths in $G$ from a single source $s$ can be found
  by running $O(\log{n})$ simple-minded Bellman-Ford steps from $s$ on $G^*$.
  Since $|E(G^*)|=\Ot\left(n+\sum_{H\in\TG(G)}|\bnd{H}|^2\right)=\Ot(n)$, this takes $\Ot(n)$ work and $O(\polylog{n})$ depth.
\end{proof}

\begin{corollary}
  The following problems on planar directed graphs have nearly work-efficient algorithms with $\Ot(n^{1/6})$ depth:
  \begin{enumerate}[label=(\arabic*)]
    \item computing a feasible flow\footnote{That is, given a vertex demand function $b:V\to\mathbb{R}$, a flow $f$
      such that the excess $e_f(v)$ of each vertex $v$ is $b(v)$.}for real-weighted capacities,
    \item computing a bipartite perfect matching,
    \item computing a maximum $s,t$-flow for polynomially bounded integral edge capacities,
    \item finding a shortest cycle going through each $e\in E(G)$,
    \item finding $s,t$-replacement paths.
  \end{enumerate}
\end{corollary}
\begin{proof}
  For the feasible flow problem, Miller and Naor~\cite{MillerN95} gave a duality-based nearly work-efficient, $O(\polylog{n})$-depth
  reduction to the real-weighted single-source shortest paths problem.

  Bipartite perfect matching is directly reducible to the feasible flow problem
  by setting the vertex demands on one side of the graph to $-1$ and on the other side to $1$.

  The maximum $s,t$-flow problem can be reduced at the cost of $O(\log{nC})$ multiplicative overhead
  to the feasible flow problem via binary search over the flow value,
  where the edge capacities are from $\mathbb{Z}\cap [0,C]$.
  This implies an $\Ot(n)$ work and $\Ot(n^{1/6})$ depth algorithm for that
  problem if $C=\poly{n}$.

  In Appendix~\ref{a:external} we describe how external dense distance graphs, defined and discussed later in this Section,
  can be used to compute the shortest cycles through all edges within desired bounds.

  Finally, Chechik and Nechushtan~\cite{ChechikN20} have recently shown that one can
  reduce the $s,t$-replacement paths problem (using a planarity-preserving near-linear time reduction)
  to computing the shortest cycles through all edges of a fixed shortest path of a graph.
  All the reduction does is basically computing a shortest $s\to t$ path $P$
  in $G$ and reversing it.
  Therefore, by Theorem~\ref{t:parallel}, the reduction can be performed in a nearly work-efficient manner
  using $\Ot(n^{1/6})$ depth.
  The final step is to compute shortest cycles through all edges
  of the reversed path $P$, which can be done by item (4).
\end{proof}

\subsection{Minimum Cost-To-Time Ratio Cycle}\label{s:meancycle}
In this section we assume that each edge $e$ of a planar digraph $G$
is assigned, besides a real weight $\wei(e)$, a time
parameter $t(e)\in \mathbb{R}_{>0}$.
Our goal is to compute a directed cycle $C\subseteq G$
minimizing the value $\lambda^*=\frac{\sum_{e\in C}\wei(e)}{\sum_{e\in C}t(e)}$.
Such a cycle is called a \emph{minimum cost-to-time ratio cycle},
or \emph{minimum ratio cycle}, in short.
In the special case when $t(e)=1$ for all $e\in E$,
$C$ is called a \emph{minimum mean cycle}.

\paragraph{Parametric search.} It is well-known that one can reduce the minimum
ratio cycle problem to the negative cycle detection problem using binary search,
as follows.
The binary search algorithm maintains
an interval $[\lambda_1,\lambda_2]$ such
that $\lambda^*\in [\lambda_1,\lambda_2]$.
Given some $\lambda\in [\lambda_1,\lambda_2]$, we wish
to decide whether $\lambda^*< \lambda$ or $\lambda\leq \lambda^*$.
Note that $\lambda^*<\lambda$ if and only if there
exists a cycle $C$ such that
$$\sum_{e\in C}w(e)-\lambda\cdot t(e)<0.$$
This condition can be clearly checked by running
a negative cycle detection algorithm on the graph~$G_\lambda$
obtained from $G$ by changing the edge weight function
to $\wei_\lambda(e):=\wei(e)-\lambda t(e)$.
By picking $\lambda=(\lambda_1+\lambda_2)/2$, this also
allows us to shrink the interval $[\lambda_1,\lambda_2]$
by half using a single negative cycle detection step.
If all the edge weights and times are
integers whose absolute values are bounded by~$W$,
the algorithm stops in $O(\log{(nW)})$ steps.
Since negative cycle detection can be solved
in $\Ot(n)$ time for planar graphs, this
leads to a weakly polynomial $\Ot(n\log{W})$-time
algorithm for the minimum ratio cycle
problem.

Megiddo's \emph{parametric search} technique
can be used to convert a strongly polynomial \emph{parallel} negative cycle detection
algorithm into a strongly polynomial time
minimum ratio cycle algorithm.
All known strongly polynomial algorithms
for this problem use variants of this technique.

Suppose we have two strongly polynomial
negative cycle detection algorithms:
(1) a parallel one~$\mathcal{P}$ with work $W(n,m)$ and depth $D(n,m)$, and
(2) a sequential one $\mathcal{S}$ with
running time $T(n,m)$.
Additionally, suppose the parallel algorithm
operates on edge weights/times (and all the stored
values dependent on the weights) only by either
additions/subtractions and comparisons.

The idea is, conceptually, to simulate (sequentially) the parallel algorithm
$\mathcal{P}$
``generically'' on all the possible graphs $G_\lambda$
with $\lambda\in [\lambda_1,\lambda_2]$ (where $[\lambda_1,\lambda_2]$ shrinks in time).
This, in particular, means that the edge weights
of $G_\lambda$ and all the stored values (e.g., the
distance labels in Bellman-Ford algorithm)
are linear functions of the form $-\lambda\cdot a+b$.
Adding and subtracting linear functions can be done
straightforwardly.
Moreover, adding or subtracting such linear functions clearly
leads to functions of the same form.
However, the result of a comparison of two
values parameterized by $\lambda$ generally depends on $\lambda$,
and different results of such a comparison might
cause different flow of $\mathcal{P}$ in the future --
we would need to handle both branches of the algorithm's flow if
we wanted to detect a negative cycle in all~$G_\lambda$.
However, this is not our goal: we only care about locating~$\lambda^*$.
So, instead, whenever $\mathcal{P}$ performs
a comparison $-\lambda a+b<-\lambda c+d$,
we compute the breakpoint $x=(d-b)/(c-a)$ (for $c\neq a$)
and use the sequential negative cycle detection
algorithm on $G_x$ to test whether $\lambda^*<x$ or $\lambda^*\geq x$.
This allows us to decide which branch would
be chosen for $\lambda=\lambda^*$
and discard the other branch
(i.e., shrink $[\lambda_1,\lambda_2]$ to either $[\lambda_1,x]$ or $[x,\lambda_2]$).
Note that this already implies a strongly
polynomial bound of $O(W(n,m)\cdot T(n,m))$ since
$\mathcal{P}$ performs $O(W(n,m))$ comparisons.

However, one can use the parallelism of $\mathcal{P}$ to do better.
In each of $O(D(n,m))$ parallel steps~$s$,~$\mathcal{P}$~performs some number $p_s$ of comparisons,
whose results depend on where $\lambda^*$ lies
relatively to some breakpoints $\lambda_1=x_0\leq x_1<\ldots<x_{p_s}\leq x_{p_s+1}=\lambda_2$,
where the sum of $p_s$ over all parallel steps $s$ is $O(W(n,m))$.
Note that the breakpoints can be sorted in $\Ot(p_s)$ time.
Afterwards, we can find such $x_i$ that
$\lambda^*\in [x_i,x_{i+1}]$
via binary search using $O(\log{(W(n,m))})$
sequential negative cycle detection runs.
Observe that this allows us to choose the correct branch
for all the comparisons in parallel step $s$ at once in just $\Ot(T(n,m))$ time.
Consequently, we obtain a (sequential) strongly
polynomial algorithm with $\Ot(W(n,m)+D(n,m)\cdot T(n,m))$ running
time.

\paragraph{Planar graphs.}
Note that all our algorithms (and therefore also the algorithm
of Lemma~\ref{l:staircase-parallel})
indeed operate
on edge weights only by performing additions
and comparisons.
Hence, by plugging the algorithm of Theorem~\ref{t:parallel}
as both the parallel and sequential algorithm into
the parametric search framework, we already
obtain an $\Ot(n^{7/6})$-time strongly polynomial minimum ratio cycle
algorithm (then, $T(n,m),W(n,m)\in \Ot(n)$
and $D(n,m)=\Ot(n^{1/6})$).

However, we can do better by slightly decreasing the depth
of the parallel algorithm of Theorem~\ref{t:parallel} at
the cost of increasing its work (we stress that
we still stick to using a sequential algorithm with $T(n,m)=\Ot(n)$).
Recall from the proof of Lemma~\ref{l:apsp-piece-parallel}
that we could actually achieve depth $\Ot(d)$
within work bounded by
$$\Ot\left(n+\sum_{H\in \TG(G)}|\bnd{H}|^3/d^3\right)=
\Ot\left(n+\frac{\sqrt{n}}{d^3}\sum_{H\in\TG(G)}|\bnd{H}|^2\right)=\Ot(n+n^{3/2}/d^3).$$
As a result, we can obtain a minimum ratio cycle algorithm
that runs in time $\Ot(dn+n^{3/2}/d^3)$ for any~$d$.
We balance these terms by choosing $d=n^{1/8}$ and obtain the following theorem.

\tmeancycle*

\subsection{Faster Computation of External DDGs}
Piecewise dense distance graphs $\DDG_H$ for $H\in\TG(G)$
have numerous other applications in sequential
planar graph algorithms beyond negative cycle detection.
Typically, though, they are computed using
the aforementioned Klein's MSSP data structure~\cite{MSSP}
in $O((|H|+|\bnd{H}|^2)\log{n})$ time rather
than inductively in $O\left((|\bnd{H_1}|^2+|\bnd{H_2}|^2)\frac{\log^2{n}}{\log^2{\log{n}}}\right)$
time using FR-Dijkstra.
However, in some situations, only the latter inductive method
can be applied.
One such application is computing so-called \emph{external dense distance graphs},
i.e., the graphs $\DDG_{G-H}$ for all $H\in \TG(G)$,
where $G-H$ is the graph obtained from $G$ by removing
the vertices $V(H)\setminus\bnd{H}$.
We set $\bnd{(G-H)}:=\bnd{H}$ since indeed $\bnd{H}$ contains
all vertices that $G-H$ shares with $H$.
One can also argue that if the faces of $H$ containing $\bnd{H}$
are simple and disjoint\footnote{Dealing with non-simple or non-disjoint faces is merely
a tedious technical difficulty (see e.g.,~\cite{exact_oracle, ItalianoKLS17, KaplanMNS17}) -- they can be avoided
by suitably extending the input graph along with its decomposition.}, $\bnd{H}$ can be assumed to lie
on $O(1)$ faces of $G-H$ as well.

It is well-known (and also not difficult to prove, see, e.g.,~\cite{MozesS12}) that if $H_2$ is a sibling of node $H_1$, whereas $H$ is $H_1$'s parent in $\TG(G)$, then $\DDG_{G-H_1}$ can be obtained by
computing all-pairs shortest paths on $\DDG_{G-H}\cup \DDG_{H_2}$
Using FR-Dijkstra (Lemma~\ref{l:fr}), this takes
$O\left((|\bnd{H}|^2+|\bnd{H_2}|^2)\frac{\log^2{n}}{\log^2{\log{n}}}\right)
=O\left((|\bnd{H_1}|^2+|\bnd{H_2}|^2)\frac{\log^2{n}}{\log^2{\log{n}}}\right)$ time and thus $O\left(n\frac{\log^3{n}}{\log^2{\log{n}}}\right)$
time through~all pieces
$H\in \TG(G)$.
On the other hand, using MSSP to compute
all the external dense distance graphs is very inefficient,
since clearly $\sum_{H\in\TG(G)}|V(G-H)|=\Omega(n^2)$.
In fact, the inductive FR-Dijkstra-based approach
has been so far the only known way to compute external dense distance graphs efficiently.
External DDGs alone can be used to compute,
for example, shortest cycles through all edges of the graph (see Appendix~\ref{a:external}).
They have also proved useful in obtaining very efficient algorithms
for maximum flow and minimum cut related problems~\cite{BorradaileSW15, LackiNSW12}.
Moreover, the computation of external dense distance graphs is the bottleneck
of construction algorithms for, e.g.,
so-called cycle-MSSP~\cite{MozesS12} data structure, or distance oracles
supporting failed vertices~\cite{failing_oracle}.

The below lemma shows that by combining the sequential version
of our APSP algorithm with FR-Dijkstra we can obtain
faster -- by a factor of almost $\Theta(\log{n})$ -- algorithm for inductively computing dense distance graphs.
For simplicity, we focus on computing $\DDG_H$ from $\DDG_{H_1}\cup\DDG_{H_2}$,
as the algorithm for computing $\DDG_{G-H}$ (as argued above) is identical.

\begin{lemma}\label{l:inductive}
 Let $H_1,H_2$ be the children of node $H\in \TG(G)$.
  Let $b=|\bnd{H_1}|+|\bnd{H_2}|$.
  We can compute all-pairs shortest paths in $\DDG_{H_1}\cup\DDG_{H_2}$
  in $O(b^2\log{n}\cdot \log\log{n}\cdot \alpha(n))$ sequential time.
  If the problem is infeasible, the algorithm
  finds a negative cycle within the same bounds.
\end{lemma}
\begin{proof}
  By Lemma~\ref{l:staircase}, in $O(b^2\cdot \alpha(n))$ time
  we can either find a negative cycle in $\DDG_{H_1}\cup\DDG_{H_2}$
  or compute a feasible price function $p$ on that graph.
  The price function allows us to make the weights
  in each $\DDG_{H_i}$ non-negative.
  Now we use the sequential version of the algorithm of Theorem~\ref{t:apsp}
  with two changes.
  First, we use a combination of Lemma~\ref{l:staircase}
  and $O(b)$ naive relaxations for each Bellman-Ford step
  performed as we did in the proof of Lemma~\ref{l:apsp-piece-parallel}.
  Even more importantly, we replace the $O((n/d)^3\log^3{n})$ time Floyd-Warshall-based
  APSP computation
  of shortest paths between all $s,t\in H_d$
  with $|H_d|$ invocations of FR-Dijkstra, as given in Lemma~\ref{l:fr}.
  As a result, the running time of the APSP
  algorithm becomes
  $$O\left(b\cdot (b\cdot \alpha(n))\cdot \log{n}\log{d}+\left(\frac{b}{d}\log{n}\right)\cdot
  b\log^2{n}\right)=O\left(b^2\log{n}\cdot \left(\alpha(n)\log{d}+\frac{\log^2{n}}{d}\right)\right).$$
  We obtain the desired bound by choosing $d=\log^2{n}$.
\end{proof}

Lemma~\ref{l:inductive} easily implies the following lemma.
\lexternal*
\bibliographystyle{plainurl}
\bibliography{references}

\newpage
\appendix

\section{Computing All-Edges Shortest Cycles}\label{a:external}

The following corollary follows easily by Lemma~\ref{l:apsp-piece-parallel}
and extending the algorithm behind Theorem~\ref{t:parallel} to also compute
external DDGs.

\begin{corollary}\label{c:external-parallel}
All external dense distance graphs $\DDG_{G-H}$ can be computed
in parallel using $\Ot(n)$ work and $\Ot(n^{1/6})$ depth.
\end{corollary}

This allows us to prove the next lemma.

\begin{lemma}\label{l:all-edges-cycles}
  Let $G$ be a real-weighted planar digraph with no negative cycles.
  Then, one can compute for each $e\in E(G)$
  the shortest cycle going through $e$:
  \begin{itemize}
    \item sequentially in $O(n\log^2{n}\cdot \log{\log{n}}\cdot \alpha(n))$ time,
    \item in parallel using $\Ot(n)$ work and $\Ot(n^{1/6})$ depth.
  \end{itemize}
\end{lemma}
\begin{proof}
Let $\TG(G)$ be the recursive decomposition of $G$.
By Lemma~\ref{l:external}, we can compute
the external distance graphs $\DDG_{G-H}$
  for all $H\in \TG(G)$ in $O(n\log^2{n}\cdot \log{\log{n}}\cdot \alpha(n))$ time.
By Corollary~\ref{c:external-parallel}, the same
  task can be completed within $\Ot(n)$ work and $\Ot(n^{1/6})$ depth.

  Let $L_e$ be any leaf node $L_e\in \TG(G)$ containing $e=uv$.
  The length of the shortest cycle going through $e$ is clearly
  $\dist_G(v,u)+\wei(e)$. So we need to compute $\dist_G(v,u)$.
  Since $\bnd{L_e}=V(L_e)$, $u,v\in \bnd{L_e}$.
  Consider the graph $G_e=\DDG_{G-L_e}\cup L_e$.
  We show that $\dist_G(v,u)=\dist_{G_e}(v,u)$.
  Indeed, consider a shortest $v\to u$ path $P=P_1\ldots P_k$ in $G$
  such that each $P_i$ is either a maximal subpath entirely contained in $G-L_e$, or a single
  edge in $L_e$. Observe that if $x_i\to y_i=P_i\subseteq G-L_e$, then by
  maximality and $u,v\in \bnd{L_e}$, we have $x_i,y_i\in \bnd{L_e}$.
  As a result, the weight of an edge $x_iy_i$ in $G_e$
  is at most $\dist_{G-L_e}(v,u)$ by the definition of $\DDG_{G-L_e}$.
  On the other hand, if $x_iy_i$ is a single edge in~$L_e$,
  then it is also preserved in $G_e$.
  As a result, a path of length $\ell(P)$ can be found in $G_e$.

  Since each $G_e$ has $O(1)$ size, computing shortest paths in all $G_e$
  takes linear extra time.
  As the computations for each $G_e$ are independent, they can
  be parallelized within $\Ot(n)$ work and $O(\polylog{n})$ depth.
\end{proof}

\end{document}